\documentclass[opre,nonblindrev]{informs3}

\OneAndAHalfSpacedXI 

\usepackage{endnotes}\usepackage{multirow}\usepackage{rotating}
\let\footnote=\endnote

\usepackage{natbib}
 \bibpunct[, ]{(}{)}{,}{a}{}{,}%
 \def\bibfont{\small}%
 \def\bibsep{\smallskipamount}%
 \def\bibhang{24pt}%

\newcommand{\V}[1]{\mbox{\boldmath $#1$}}

\newcommand{\OPT}{{\mbox{OPT}}}

\newcommand{\LP}{{\tt{LP}}}

\newcommand{\ONLINE}{\mbox{ONLINE}}

\newcommand{\OBJ}{\widehat{\mbox{OPT}}}

\newcommand{\topw}{\mbox{topw}}

\newcommand{\Ex}{\mathbb{E}}

\newcommand{\comment}[1]{}
\newcommand{\shcomment}[1]{}
\definecolor{light-gray}{gray}{0.4}

\newcommand{\ba}{\mbox{\boldmath $a$}}

\newcommand{\bff}{\mbox{\boldmath $f$}}
\newcommand{\bg}{\mbox{\boldmath $g$}}

\newcommand{\bp}{\mbox{\boldmath $p$}}
\newcommand{\bq}{\mbox{\boldmath $q$}}

\newcommand{\bx}{\mbox{\boldmath $x$}}
\newcommand{\by}{\mbox{\boldmath $y$}}

\TheoremsNumberedThrough     
\ECRepeatTheorems

\EquationsNumberedThrough    

\begin{document}


\RUNAUTHOR{Agrawal, Wang and Ye}

\RUNTITLE{A Dynamic Near-Optimal Algorithm for Online Linear
Programming}

\TITLE{A Dynamic Near-Optimal Algorithm for Online Linear
Programming}

\ARTICLEAUTHORS{ \AUTHOR{Shipra Agrawal}\AFF{Microsoft Research
India, Bangalore, India, \EMAIL{shipra@microsoft.com}}
\AUTHOR{Zizhuo Wang} \AFF{Department of Industrial and Systems
Engineering, University of Minnesota, Minneapolis, MN 55455,
\EMAIL{zwang@umn.edu}} \AUTHOR{Yinyu Ye} \AFF{Department of
Management Science and Engineering, Stanford University, Stanford,
CA 94305, \EMAIL{yinyu-ye@stanford.edu}} }

\ABSTRACT{A natural optimization model that formulates many online
resource allocation problems is the online linear programming (LP)
problem in which the constraint matrix is revealed column by column
along with the corresponding objective coefficient. In such a model,
a decision variable has to be set each time a column is revealed
without observing the future inputs and the goal is to maximize the
overall objective function. In this paper, we propose a near-optimal
algorithm for this general class of online problems under the
assumptions of random order of arrival and some mild conditions on
the size of the LP right-hand-side input. Specifically, our
learning-based algorithm works by dynamically updating a threshold
price vector at geometric time intervals, where the dual prices
learned from the revealed columns in the previous period are used to
determine the sequential decisions in the current period. Due to the
feature of dynamic learning, the competitiveness of our algorithm
improves over the past study of the same problem. We also present a
worst-case example showing that the performance of our algorithm is
near-optimal.}



\maketitle

%


\section{Introduction}
\label{sec:intro} Online optimization is attracting increasingly
wide attention in the computer science, operations research, and
management science communities. In many practical problems, data
does not reveal itself at the beginning, but rather comes in an
online fashion. For example, in online revenue management problems,
consumers arrive sequentially, each requesting a subset of goods
(e.g., multi-leg flights or a period of stay in a hotel) and
offering a bid price. On observing a request, the seller needs to
make an irrevocable decision whether to accept or reject the current
bid with the overall objective of maximizing the revenue while
satisfying the resource constraints. Similarly, in online routing
problems, the network capacity manager receives a sequence of
requests from users with intended usage of the network, each with a
certain utility. And his objective is to allocate the network
capacity to maximize the social welfare. A similar format also
appears in online auctions, online keyword matching problems, online
packing problems, and various other online revenue management and
resource allocation problems. For an overview of the online
optimization literature and its recent development, we refer the
readers to \cite{borodin}, \cite{buchbiner-long09} and
\cite{Devanur2011}.

In many examples mentioned above, the problem can be formulated as
an online linear programming problem (Sometimes, people consider the
corresponding integer program. While our discussion focuses on the
linear programming relaxation of these problems, our results
naturally extend to integer programs. See Section
\ref{sec:ext-integer} for the discussion on this). An online linear
programming problem takes a linear program as its underlying form,
while the constraint matrix is revealed column by column with the
corresponding coefficient in the objective function. After observing
the input arrived so far, an immediate decision must be made without
observing the future data. To be precise, we consider the following
(offline) linear program
\begin{equation}
\begin{array}{lll}\label{primaloffline}
\mbox{maximize}&\sum_{j=1}^n \pi_jx_j\\
\mbox{subject to} & \sum_{j=1}^n a_{ij}x_{j}\le b_i, & i=1,\ldots, m\\
&0\le x_j\le 1,& j=1,\ldots, n,\\
\end{array}
\end{equation}
where for all $j$, $\pi_j\ge0$, $\V{a}_j=\{a_{ij}\}_{i=1}^m\in
[0,1]^m$,\footnote{The assumption that $a_{ij}\le 1$ is not
restrictive as we can normalize the constraint to meet this
requirement.} and $\V{b}=\{b_i\}_{i=1}^m\in {\mathbb R}^m$. In the
corresponding {\it online linear programming} problem, at each time
$t$, the coefficients $(\pi_t, \V{a}_t)$ are revealed, and the
decision variable $x_t$ has to be chosen. Given the previous $t-1$
decisions $x_1, \ldots, x_{t-1}$, and input $\{\pi_j,
\V{a}_j\}_{j=1}^t$ until time $t$, the $t^{th}$ decision variable
$x_t$ has to satisfy
\begin{equation}
\begin{array}{ll}\label{model1}
\sum_{j=1}^ta_{ij}x_j\le b_i, & i=1,\ldots, m\\
 0\le x_t\le 1.
\end{array}
\end{equation}
The goal in the online linear programming problem is to choose
$x_t$'s such that the objective function $\sum_{t=1}^{n}\pi_tx_t$ is
maximized.

In this paper, we propose algorithms that achieve good performance
for solving the online linear programming problem. In order to
define the performance of an algorithm, we first need to make some
assumptions regarding the input parameters. We adopt the following
random permutation model in this paper:
\begin{assumption}\label{assumption1}
The columns $\V{a}_j$ (with the objective coefficient $\pi_j$)
arrive in a random order. The set of columns
$(\V{a}_1,\V{a}_2,...,\V{a}_n)$ can be adversarily picked at the
start. However, the arrival order of $(\V{a}_1,\V{a}_2,...,\V{a}_n)$
is uniformly distributed over all the permutations.
\end{assumption}
\begin{assumption}\label{assumption2}
We know the total number of columns $n$ a priori.
\end{assumption}

The random permutation model has been adopted in many existing
literature for online problems (see Section
\ref{subsection:literature_review} for a comprehensive review of the
related literature). It is an intermediate path between using a
worst-case analysis and assuming the distribution of the input is
known. The worst-case analysis, which is completely robust to input
uncertainty, evaluates an algorithm based on its performance on the
worst-case input (see, e.g., \cite{mehta}, \cite{buchbinder-or09}).
However, this leads to very pessimistic performance bounds for this
problem: no online algorithm can achieve better than $O(1/n)$
approximation of the optimal offline solution (\cite{babaioff}). On
the other hand, although a priori input distribution can simplify
the problem to a great extent, the choice of distribution is very
critical and the performance can suffer if the actual input
distribution is not as assumed. Specifically, Assumption 1 is weaker
than assuming the columns are drawn independently from some
(possibly unknown) distribution. Indeed, one can view $n$ i.i.d.
columns as first drawing $n$ samples from the underlying
distribution and then randomly permute them. Therefore, our proposed
algorithm and its performance would also apply if the input data is
drawn i.i.d. from some distribution.

Assumption 2 is required since we need to use the quantity $n$ to
decide the length of history for learning the threshold prices in
our algorithm. In fact, as shown in \cite{Devanur}, it is necessary
for any algorithm to get a near-optimal performance.\footnote{An
example to show the knowledge of $n$ is necessary to obtain a
near-optimal algorithm is as follows. Suppose there is only one
product and the inventory is $n$. And all $a_i$'s are $1$. There
might be $n$ or $2n$ arrivals. And in either case, half of them have
value $1$ and half of them have value $2$. Now consider any
algorithm. If it accepts less than $2/3$ among the first $n$
arrivals, the loss is at least $n/3$ (or $1/6$ of the optimal value)
if in fact there are $n$ arrivals in total. On the other hand, if it
accepts more than $2/3$ among the first $n$ arrivals, then it must
have accepted more than $n/6$ bids with value $1$. And if the true
number of arrivals is $2n$, then it will also have a loss of at
least $1/12$ of the true optimal value. Thus if one doesn't know the
exact $n$, there always exists a case where the loss is a constant
fraction of the true optimal value.} However, this assumption can be
relaxed to an approximate knowledge of $n$ (within at most $1 \pm
\epsilon$ multiplicative error), without affecting our results.

We define the competitiveness of online algorithms as follows:

\begin{definition} Let $\OPT$ denote the optimal objective value for
the offline problem (\ref{primaloffline}). An online algorithm
${\cal A}$ is {\it $c$-competitive} in the random permutation model
if the expected value of the online solution obtained by using
${\cal A}$ is at least $c$ factor of the optimal offline solution.
That is,
\begin{eqnarray*}
\Ex_{\sigma}\left[\sum_{t=1}^n \pi_tx_t\right] \ge c \cdot \OPT,
\end{eqnarray*}
where the expectation is taken over uniformly random permutations
$\sigma$ of $1, \ldots, n$, and $x_t$ is the $t^{th}$ decision made
by algorithm ${\cal A}$ when the inputs arrive in order $\sigma$.
\end{definition}

In this paper, we present a near-optimal algorithm for the online
linear program (\ref{model1}) under the above two assumptions and a
lower bound condition on the size of $\V{b}$. We also extend our
results to the following more general online linear optimization
problems with multi-dimensional decisions at each time period:
\begin{itemize}
\item Consider a sequence of $n$
non-negative vectors $\V{f}_1, \V{f}_2, \ldots, \V{f}_n\in
\mathbb{R}^k$, $mn$ non-negative vectors
\[\V{g}_{i1}, \V{g}_{i2}, \ldots, \V{g}_{in}\in [0,1]^k,\quad i=1,\ldots, m,\]
and $K=\{\bx \in {\mathbb R}^k: \bx^T\V{e}\le 1, \bx \ge 0\}$ (we
use $\V{e}$ to denote the all 1 vectors). The offline linear program
is to choose $\bx_1$,...,$\bx_n$ to solve
\begin{equation*}
\begin{array}{lll}
\label{multidimensionaloffline} \mbox{maximize} & \sum_{j=1}^n
\bff_j^T\bx_j \\
\mbox{subject to} & \sum_{j=1}^n \bg_{ij}^T\bx_j \le b_i, & i =
1,...,m\\
& \bx_j \in K.
\end{array}
\end{equation*}
In the corresponding online problem, given the previous $t-1$
decisions $\bx_1, \ldots, \bx_{t-1}$, each time we choose a
$k$-dimensional decision $\bx_t\in {\mathbb R}^k$, satisfying:
\begin{equation}
\begin{array}{rl}
\label{multidimensionalonline}
\sum_{j=1}^t \V{g}^T_{ij}\bx_j \le b_i, & i=1, \ldots, m\\
 \bx_t \in K,\\
\end{array}
\end{equation}
using the knowledge up to time $t$. And the objective is to maximize
$\sum_{j=1}^n \V{f}^T_j \bx_j$ over the entire time horizon. Note
that Problem (\ref{model1}) is a special case of Problem
(\ref{multidimensionalonline}) with $k=1$.
\end{itemize}

\subsection{Specific applications}
In the following, we show some specific applications of the online
linear programming model. The examples are only a few among the wide
range of applications of this model.

\subsubsection{Online knapsack/secretary problem} The one
dimensional version of the online linear programming problem studied
in this paper is usually referred as online knapsack or secretary
problem. In such problems, a decision maker faces a sequence of
options, each with a certain cost and a value, and he has to choose
a subset of them in an online fashion so as to maximize the total
value without violating the cost constraint. Applications of this
problem arise in many contexts, such as hiring workers, scheduling
jobs and bidding in sponsored search auctions.

Random permutation model has been widely adopted in the study of
this problem, see \cite{kleinberg} and \cite{babaioff-knapsack} and
references thereafter. In those papers, either a constant
competitive ratio is obtained for finite-sized problems or a
near-optimal algorithm is proposed for large-sized problems. In this
paper, we study an extension of this problem to higher dimension and
propose a near-optimal algorithm for it.

\subsubsection{Online routing problem}
Consider a computer network connected by $m$ edges, each edge $i$
has a bounded capacity (bandwidth) $b_i$. There are a large number
of requests arriving online, each asking for certain capacities
$\V{a}_t \in {\mathbb R}^m$ in the network, along with a utility or
price for his/her request. The offline problem for the decision
maker is given by the following integer program:
\begin{equation*}
\label{routing}
\begin{array}{lll}
\mbox{maximize }& \sum_{t=1}^n \pi_t x_t & \\
\mbox{subject to} & \sum_{t=1}^n a_{it}x_t \le b_i & i=1,\ldots, m \\
& x_t \in \{0,1\}. &
\end{array}
\end{equation*}
Discussions of this problem can be found in \cite{buchbinder-or09},
\cite{awerbuch} and references therein. Note that this problem is
also studied under the name of online packing problem.

\subsubsection{Online adwords problem}
Selling online advertisements has been the main revenue driver for
many internet companies such as Google, Yahoo, etc. Therefore
improving the performance of ads allocation systems becomes
extremely important for those companies and thus has attracted great
attention in the research community in the past decade. In the
literature, the majority of the research adopts an {\it online
matching} model, see e.g., \cite{mehta, Goel, Devanur, Karande,
Bahmani, mahdian, feldman2010online, feldman2009, feldman2009focs}.
In such models, there are $n$ search queries arriving online. And
there are $m$ bidders (advertisers) each with a daily budget $b_i$.
Based on the relevance of each search keyword, the $i$th bidder will
bid a certain amount $\pi_{ij}$ on query $j$ to display his
advertisement along with the search result.\footnote{Here we assume
the search engines use a pay-per-impression scheme. The model can be
easily adapted to a pay-per-click scheme by multiplying the bid
value by the click-through-rate parameters. Also we assume there is
only one advertisement slot for each search result.} For the
$j^{th}$ query, the decision maker (i.e., the search engine) has to
choose an $m$-dimensional vector $\bx_j=\{x_{ij}\}_{i=1}^m$, where
$x_{ij} \in \{0,1\}$ indicates whether the $j^{th}$ query is
allocated to the $i^{th}$ bidder. The corresponding offline problem
can be formulated as:
\begin{equation}
\begin{array}{lll}\label{adword}
\mbox{maximize } & \sum_{j=1}^n \V{\pi}^T_j\V{x}_j & \\
\mbox{subject to} & \sum_{j=1}^n \pi_{ij}x_{ij} \le b_i, & i=1,\dots,m \\
& \bx_j^T\V{e} \le 1 \\
& \bx_j \in \{0, 1\}^m.
\end{array}
\end{equation}
The linear programming relaxation of (\ref{adword}) is a special
case of the general online linear programming problem
(\ref{multidimensionalonline}) with $\V{f}_j = \V{\pi}_j$,
$\V{g}_{ij}=\pi_{ij}\V{e}_i$ where $\V{e}_i$ is the $i$th unit
vector of all zeros except $1$ for the $i$th entry.

In the literature, the random permutation assumption has attracted
great interests recently for its tractability and generality.
Constant competitive algorithm as well as near-optimal algorithms
have been proposed. We will give a more comprehensive review in
Section \ref{subsection:literature_review}.

\subsection{Key ideas and main results}
\label{subsection:main results}

The main contribution of this paper is to propose an algorithm that
solves the online linear programming problem with a near-optimal
competitive ratio under the random permutation model. Our algorithm
is based on the observation that the optimal solution ${\bx}^*$ for
the offline linear program can be largely determined by the optimal
dual solution $\bp^* \in \mathbb{R}^m$ corresponding to the $m$
inequality constraints. The optimal dual solution acts as a {\it
threshold price} so that $x^*_j>0$ only if $\pi_j\ge{\bp^*}^T
\V{a}_j $. Our online algorithm works by learning a threshold price
vector from some initial inputs. The price vector is then used to
determine the decisions for later periods. However, instead of
computing the price vector only once, our algorithm initially waits
until $\epsilon n$ steps or arrivals, and then computes a new price
vector every time the history doubles, i.e., at time $\epsilon n,
2\epsilon n, 4\epsilon n, \ldots$ and so on. We show that our
algorithm is $1-O(\epsilon)$-competitive in the random permutation
model under a size condition of the right-hand-side input. Our main
results are precisely stated as follows:

\begin{theorem}
\label{th:main} For any $\epsilon>0$, our online algorithm is
$1-O(\epsilon)$ competitive for the online linear program
(\ref{model1}) in the random permutation model, for all inputs such
that
\begin{equation}\label{condition2}
B=\min_ib_i\ge\Omega\left(\frac{m\log{(n/\epsilon)}}{\epsilon^2}\right).
\end{equation}
\end{theorem}

An alternative way to state Theorem \ref{th:main} is that our
algorithm has a competitive ratio of $1-
O\left(\sqrt{m\log{n}/B}\right)$. We prove Theorem \ref{th:main} in
Section \ref{sec:dynamic}. Note that the condition in Theorem
\ref{th:main} depends on $\log{n}$, which is far from satisfying
everyone's demand when $n$ is large. In \cite{kleinberg}, the author
proves that $B\ge 1/\epsilon^2$ is necessary to get a
$1-O(\epsilon)$ competitive ratio in the $B$-secretary problem,
which is the single dimensional counterpart of the online LP problem
with $a_{t}=1$ for all $t$. Thus, the dependence on $\epsilon$ in
Theorem \ref{th:main} is near-optimal. In the next theorem, we show
that a dependence on $m$ is necessary for any online algorithm to
obtain a near-optimal solution. Its proof will appear in Section
\ref{sec:low}.

\begin{theorem}
\label{th:lowerbound} For any algorithm for the online linear
programming problem (\ref{model1}) in the random permutation model,
there exists an instance such that its competitive ratio is less
than $1-\Omega(\epsilon)$ when
$$ B=\min_i b_i \le \frac{\log(m)}{\epsilon^2}.$$
Or equivalently, no algorithm can achieve a competitive ratio better
than $1-\Omega\left(\sqrt{\log{m}/B}\right)$.
\end{theorem}

We also extend our results to the more general model as introduced
in (\ref{multidimensionalonline}) :
\begin{theorem}
\label{th:convex} For any $\epsilon>0$, our algorithm is
$1-O(\epsilon)$ competitive for the general online linear
programming problem (\ref{multidimensionalonline}) in the random
permutation model, for all inputs such that:
\begin{equation}
\label{convexcondition} B=\min_ib_i \ge
\Omega\left(\frac{m\log{(nk/\epsilon)}}{\epsilon^2}\right).
\end{equation}
\end{theorem}

Now we make some remarks on the conditions in Theorem \ref{th:main}
and \ref{th:convex}. First of all, the conditions only depend on the
right-hand-side input $b_i$'s, and are independent of the size of
$\OPT$ or the objective coefficients. And by the random permutation
model, they are also independent of the distribution of the input
data. In this sense, our results are quite robust in terms of the
input data uncertainty. In particular, one advantage of our result
is that the conditions are checkable before the algorithm is
implemented, which is unlike the conditions in terms of $\OPT$ or
the objective coefficients.  Even just in terms of $b_i$, as shown
in Theorem \ref{th:lowerbound}, the dependence on $\epsilon$ is
already optimal and the dependence on $m$ is necessary. Regarding
the dependence on $n$, we only need $B$ to be of order $\log{n}$,
which is far less than the total number of bids $n$. Indeed, the
condition might be strict for some small-sized problems. However, if
the budget is too small, it is not hard to imagine that no online
algorithm can do very well. On the contrary, in applications with
large amount of inputs (for example, in the online adwords problem,
it is estimated that a large search engine could receive several
billions of searches per day, even if we focus on a specific
category, the number can still be in the millions) with reasonably
large right-hand-side inputs (e.g., the budgets for the advertiser),
the condition is not hard to satisfy. Furthermore, the conditions in
Theorem \ref{th:main} and \ref{th:convex} are just theoretical
results, the performance of our algorithm might still be very good
even if the conditions are not satisfied (as shown in some numerical
tests in \cite{wang_thesis}). Therefore, our results are of both
theoretical and practical interests.

Finally, we finish this section with the following corollary:

\begin{corollary}\label{cor:nonuniforma}
In the online linear programming problem (\ref{model1}) and
(\ref{multidimensionalonline}), if the largest entry of constraint
coefficients does not equal to $1$, then both our Theorem
\ref{th:main} and \ref{th:convex} still hold with the conditions
(\ref{condition2}) and (\ref{convexcondition}) replaced by:
\[\frac{b_i}{\bar{a}_i}\ge \Omega\left(\frac{m\log{(nk/\epsilon)}}{\epsilon^2}\right),\ \forall i,
\]
where, for each row $i$, $\bar{a}_i=\max_j\{|a_{ij}|\}$ of
(\ref{model1}), or $\bar{a}_i=\max_j\{\|\V{g}_{ij}\|_{\infty}\}$ of
(\ref{multidimensionalonline}).
\end{corollary}

\subsection{Related work}
\label{subsection:literature_review} The design and analysis of
online algorithms have been a topic of wide interest in the computer
science, operations research, and management science communities.
Recently, the random permutation model has attracted growing
popularity since it avoids the pessimistic lower bounds of the
adversarial input model while still capturing the uncertainty of the
inputs. Various online algorithms have been studied under this
model, including the secretary problem (\cite{kleinberg, babaioff}),
the online matching and adwords problem (\cite{Devanur,
feldman2009focs, Goel, mahdian, Karande, Bahmani}) and the online
packing problem (\cite{feldman2010online, Molinaro2012}). Among
these work, two types of results are obtained: one achieves a
constant competitive ratio independent of the input parameters; the
other focuses on the performance of the algorithm when the input
size is large. Our paper falls into the second category. In the
following literature review, we will focus ourselves on this
category of work.

The first result that achieves a near-optimal performance in the
random permutation model is by \cite{kleinberg}, in which a
$1-O(1/\sqrt{B})$ competitive algorithm is proposed for the single
dimensional multiple-choice secretary problem. The author also
proves that the $1-O(1/\sqrt{B})$ competitive ratio achieved by his
algorithm is the best possible for this problem. Our result extends
his work to mutli-dimensional case with competitiveness
$1-O(\sqrt{m\log{n}/B})$. Although the problem looks similar, due to
the multi-dimensional structure, different algorithms are needed and
different techniques are required for our analysis. Specifically,
\cite{kleinberg} recursively applies a randomized version of the
classical secretary algorithm while we maintain a price based on the
linear programming duality theory and have a fixed price updating
schedule. We also prove that no online algorithm can achieve a
competitive ratio better than $1-\Omega(\sqrt{\log{m}/B})$ for the
multi-dimensional problem. To the best of our knowledge, this is the
{\it first} result that shows the necessity of dependence on the
dimension $m$, for the best competitive ratio achievable for this
problem. It clearly points out that high dimensionality indeed adds
to the difficulty of this problem.

Later, \cite{Devanur} study a linear programming based approach for
the online adwords problem. In their approach, they solve a linear
program {\emph{once}} and utilize its dual solution as threshold
price to make future decisions. The authors prove a competitive
ratio of $1-O(\sqrt[3]{\pi_{\max}m^2\log{n}/\OPT})$ for their
algorithm. In our work, we consider a more general model and develop
an algorithm which updates the dual prices at a carefully chosen
pace. By using {\emph {dynamic}} updates, we achieve a competitive
ratio that can depend only on $B$: $1-O(\sqrt{m\log{n}/B})$. This is
attractive in practice since $B$ can be checked before the problem
is solved while $\OPT$ can not. Moreover, we show that the
dependence on $B$ of our result is optimal. Although our algorithm
shares similar ideas to theirs, the dynamic nature of our algorithm
requires a much more delicate design and analysis. We also answer
the important question of how often we should update the dual prices
and we show that significant improvements can be made by using the
dynamic learning algorithm.

Recently, \cite{feldman2010online} study a more general online
packing problem which allows the dimension of the choice set to vary
at each time period (a further extension of
(\ref{multidimensionalonline})). They propose a one-time learning
algorithm which achieves a competitive ratio that depends both on
the right-hand-side $B$ and $\OPT$. And the dependence on $B$ is of
order $1-O(\sqrt[3]{m\log{n}/B})$. Therefore, comparing to their
competitive ratio, our result not only removes the dependence on
$\OPT$, but also improves the dependence on $B$ by an order.  We
show that the improvement is due to the use of dynamic learning.

More recently, \cite{Molinaro2012} study the same problem and obtain
a competitive ratio of $1-O(\sqrt{m^2 \log{m}/B})$. The main
structure of their algorithm (especially the way they obtain square
root rather than cubic root) is modified from that in this paper.
They further use a novel covering technique to remove the dependence
on $n$ in the competitive ratio, at an expense of increasing an
order of $m$. In contrast, we present the improvement from the cubic
root to square root and how to remove the dependence on $\OPT$.

A comparison of the results of \cite{kleinberg}, \cite{Devanur},
\cite{feldman2010online}, \cite{Molinaro2012} and this work is shown
in Table \ref{tab:summary_of_result}.

\begin{table*}[t]
\begin{center}
\begin{tabular}{|c | c |}
\hline
& Competitiveness \\
\hline
Kleinberg (2005)  & $1 - O\left(1/\sqrt{B}\right) (\mbox{only for } m = 1)$                   \\
Devanur and Hayes (2009)  & $1-O(\sqrt[3]{\pi_{\max}m^2\log{n}/\OPT})$                 \\
Feldman et al. (2010) & $1-O(\max\{\sqrt[3]{m\log{n}/B}, \pi_{\max}m\log{n}/OPT\})$  \\
Molinaro and Ravi (2012) & $1-O(\sqrt{m^2\log{m}/B})$\\
This paper  & $1 - O(\sqrt{m\log{n}/B})$  \\
\hline
\end{tabular}
\caption{Comparison of existing
results}\label{tab:summary_of_result}
\end{center}
\end{table*}

Besides the random permutation model, \cite{Devanur2011fast} study
an online resource allocation problem under what they call the
adversarial stochastic input model. This model generalizes the case
when the columns are drawn from an i.i.d. distribution, however, it
is more stringent than the random permutation model. In particular,
their model does not allow the situations when there might be a
number of ``shocks'' in the input series. For this input model, they
develop an algorithm that achieves a competitive ratio of
$1-O\left(\max\{\sqrt{\log{m}/B},
\sqrt{\pi_{max}\log{m}/{\OPT}}\}\right)$. Their result is
significant in that it achieves near-optimal dependence on $m$.
However, the dependence on $\OPT$ and the stronger assumption makes
it not directly comparable to our results. And their algorithm uses
quite different techniques from ours.

In the operations research and management science communities, a
dynamic and optimal pricing strategy for various online revenue
management and resource allocation problems has always been an
important research topic, some literature include \cite{Elmaghraby}
,\cite{gallego,gallego2}, \cite{Talluri}, \cite{cooper} and
\cite{bitran}. In \cite{gallego,gallego2} and \cite{bitran}, the
arrival processes are assumed to be price sensitive. However, as
commented in \cite{cooper}, this model can be reduced to a price
independent arrival process with availability control under Poisson
arrivals. Our model can be further viewed as a discrete version of
the availability control model which is also used as an underlying
model in \cite{Talluri} and discussed in \cite{cooper}. The idea of
using a threshold - or ``bid'' - price is not new. It is initiated
in \cite{williamson, simpson} and investigated further in
\cite{Talluri}. In \cite{Talluri}, the authors show that the bid
price control policy is asymptotically optimal. However, they assume
the knowledge on the arrival process and therefore the price is
obtained by ``forecasting'' the future using the distribution
information rather than ``learning'' from the past observations as
we do in our paper. The idea of using linear programming to find the
dual optimal bid price is discussed in \cite{cooper} where
asymptotic optimality is also achieved. But again, the arrival
process is assumed to be known which makes the analysis quite
different.

The contribution of this paper is several fold. First, we study a
general online linear programming framework, extending the scope of
many prior work. And due to its dynamic learning capability, our
algorithm is {\it distribution free}--no knowledge on the input
distribution is assumed except for the random order of arrival and
the total number of entries. Moreover, instead of learning the price
just once, we propose a {\it dynamic} learning algorithm that
updates the prices as more information is revealed. The design of
such an algorithm answers the question raised in \cite{cooper}, that
is, how often and when should one update the price? We give an
explicit answer to this question by showing that updating the prices
at geometric time intervals -not too sparse nor too often- is
optimal. Thus we present a precisely quantified strategy for dynamic
price update. Furthermore, we provide a non-trivial lower bound for
this problem, which is the first of its kind and show for the first
time that the dimensionality of the problem adds to its difficulty.

In our analysis, we apply many standard techniques from Probably
Approximately Correct learning (PAC-learning), in particular,
concentration bounds and covering arguments. Our dynamic learning
also shares a similar idea as the ``doubling trick'' used in
learning problems. However, unlike the doubling trick which is
typically applied to an unknown time horizon (\cite{Cesa}), we show
that a geometric pace of price updating in a fixed length of horizon
with a careful design could also enhance the performance of the
algorithm.

\subsection{Organization}
\label{subsec:organization} The rest of the paper is organized as
follows. In Section \ref{sec:one-time} and \ref{sec:dynamic}, we
present our online algorithm and prove that it achieves
$1-O(\epsilon)$ competitive ratio under mild conditions on the
input. To keep the discussions clear and easy to follow, we start in
Section \ref{sec:one-time} with a simpler one-time learning
algorithm. While the analysis for this simpler algorithm will be
useful to demonstrate our proof techniques, the results obtained in
this setting are weaker than those obtained by our dynamic learning
algorithm, which is discussed in Section \ref{sec:dynamic}. In
Section \ref{sec:low}, we give a detailed proof of Theorem
\ref{th:lowerbound} regarding the necessity of lower bound condition
used in our main theorem. In Section \ref{sec:ext}, we present
several extensions of our study. Then we conclude our paper in
Section \ref{sec:conclusions}.

\section{One-time Learning Algorithm}
\label{sec:one-time}

In this section, we propose a one-time learning algorithm for the
online linear programming problem. We consider the following partial
linear program defined only on the input until time $s= \epsilon n $
(for the ease of notation, without loss of generality, we assume
$\epsilon n $ is an integer throughout our analysis):
\begin{equation}
\begin{array}{lll}\label{primalsample}
\mbox{maximize} &\sum_{t=1}^{s} \pi_t x_t & \\
\mbox{subject to}& \sum_{t=1}^{s} a_{it}x_t\le
(1-\epsilon)\frac{s}{n} b_i, & i=1,\ldots, m\\
& 0 \le x_t \le 1, & t=1, \ldots, s,\\
\end{array}
\end{equation}
and its dual problem:
\begin{equation}
\begin{array}{lll}\label{dualsample}
\mbox{minimize} & \sum_{i=1}^m (1-\epsilon)\frac{s}{n}b_ip_i
+ \sum_{t=1}^s y_t &\\
\mbox{subject to} & \sum_{i=1}^m a_{it}p_i + y_t \ge \pi_t, & t =
1,\ldots, s\\
& p_i, y_t \ge 0, & i = 1,\ldots, m, t=1,\ldots, s.\\
\end{array}
\end{equation}
Let $(\V{\hat{p}}, \V{\hat{y}})$ be the optimal solution to
(\ref{dualsample}). Note that $\V{\hat{p}}$ has the natural meaning
of the price for each resource. For any given price vector $\V{p}$,
we define the allocation rule $x_t(\bp)$ as follows:
\begin{equation}
\label{linear:xp} x_t(\bp) = \left\{\begin{array}{ll}
                    0 & \mbox{ if } \pi_t \le \bp^T\V{a}_t\\
                    1 & \mbox{ if } \pi_t   > \bp^T\V{a}_t.\\
\end{array}
\right.
\end{equation}

We now state our one-time learning algorithm:\\

{\bf\noindent Algorithm OLA (One-time Learning Algorithm):}
\begin{enumerate}
\item Initialize $x_t = 0$, for all
$t\le s$. And $\hat{\V{p}}$ is defined as above.
\item For $t=s+1, s+2, \dots,n$, if $a_{it}x_t(\hat{\bp})\le b_i-\sum_{j=1}^{t-1}a_{ij}x_j$ for all $i$,
set $x_t=x_t(\hat{\bp})$; otherwise, set $x_t=0$. Output $x_t$.\\
\end{enumerate}
In the one-time learning algorithm, we learn a dual price vector
using the first $\epsilon n$ arrivals. Then, at each time $t>
\epsilon n$, we use this dual price to decide the current
allocation, and execute this decision as long as it doesn't violate
any of the constraints. An attractive feature of this algorithm is
that it requires to solve only one small linear program, defined on
$\epsilon n$ variables. Note that the right-hand-side of
(\ref{primalsample}) is modified by a factor $1-\epsilon$. This
modification is to guarantee that with high probability, the
allocation $x_t(\V{p})$ does not violate the constraints. This trick
is also used in Section \ref{sec:dynamic} when we study the dynamic
learning algorithm. In the next subsection, we prove the following
proposition regarding the competitive ratio of the one-time learning
algorithm, which relies on a stronger condition than Theorem
\ref{th:main}:
\begin{proposition}
\label{prop:one-time} For any $\epsilon>0$, the one-time learning
algorithm is $1-6\epsilon$ competitive for the online linear program
(\ref{model1}) in the random permutation model, for all inputs such
that
\begin{equation*}
\label{condition-one-time} B=\min_ib_i\ge
\frac{6m\log(n/\epsilon)}{\epsilon^3}.
\end{equation*}
\end{proposition}

\subsection{Competitive Ratio Analysis}
\label{sec:comptetive_ratio_analysis1}

Observe that the one-time learning algorithm waits until time $s=
\epsilon n $, and then sets the solution at time $t$ as
$x_t(\hat{\bp})$, unless it violates the constraints. To prove its
competitive ratio, we follow the following steps. First we show that
if $\bp^*$ is the optimal dual solution to (\ref{primaloffline}),
then $\{x_t(\bp^*)\}$ is close to the primal optimal solution
$\V{x}^*$, i.e., learning the dual price is sufficient to determine
a close primal solution. However, since the columns are revealed in
an online fashion, we are not able to obtain $\bp^*$ during the
decision period. Instead, in our algorithm, we use $\hat{\bp}$ as a
substitute. We then show that $\hat{\bp}$ is a good substitute of
${\bp}^*$: 1) with high probability, $x_t(\hat{\bp})$ satisfies all
the constraints of the linear program; 2) the expected value of
$\sum_t \pi_t x_t(\hat{\bp})$ is close to the optimal offline value.
Before we start our analysis, we make the following simplifying
technical assumption in our discussion:

\begin{assumption}
\label{assumption:ties} The problem inputs are in general position,
namely for any price vector $\bp$, there can be at most $m$ columns
such that $\bp^T\V{a}_t=\pi_t$.
\end{assumption}
Assumption \ref{assumption:ties} is not necessarily true for all
inputs. However, as pointed out by \cite{Devanur}, one can always
randomly perturb $\pi_t$ by arbitrarily small amount $\eta$ through
adding a random variable $\xi_{t}$ taking uniform distribution on
interval $[0,\eta]$. In this way, with probability 1, no $\bp$ can
satisfy $m+1$ equations simultaneously among $\bp^T\V{a}_t=\pi_t$,
and the effect of this perturbation on the objective can be made
arbitrarily small. Under this assumption, we can use the
complementarity conditions of linear program (\ref{primaloffline})
to obtain the following lemma.

\begin{lemma}
\label{KKT} $x_t(\bp^*)\le x^*_t$ for all $t$, and under Assumption
\ref{assumption:ties}, $x^*_t$ and $x_t(\bp^*)$ differs for no more
than $m$ values of $t$.
\end{lemma}
{\noindent\bf Proof.} Consider the offline linear program
(\ref{primaloffline}) and its dual (let $\bp$ denote the dual
variables associated with the first set of constraints and $y_t$
denote the dual variables associated with the constraints $x_t \le
1$):
\begin{equation}
\begin{array}{lll}\label{mainmodeldual}
\mbox{minimize} & \sum_{i=1}^m b_ip_i + \sum_{t=1}^n y_t \\
\mbox{subject to} & \sum_{i=1}^m a_{it}p_i + y_t \ge \pi_t, & t = 1,...,n\\
& p_i, y_t \ge 0, & i = 1,...,m,t = 1,...,n.
\end{array}
\end{equation}
By the complementarity slackness conditions, for any optimal
solution $\bx^*$ for the primal problem (\ref{primaloffline}) and
optimal solution $(\bp^*, \by^*)$ for the dual, we must have:
\begin{equation*}
x_t^*\cdot\left(\sum_{i=1}^m a_{it}p_i^* + y_t^* - \pi_t\right) = 0
\quad\quad\mbox{and}\quad\quad (1-x_t^*) \cdot y_t^* = 0 \quad\quad
\mbox{ for all } t.
\end{equation*}
If $x_t(\bp^*) = 1$, by (\ref{linear:xp}), $\pi_t > (\bp^*)^T\ba_t$.
Thus, by the constraint in (\ref{mainmodeldual}), $y^*_t > 0$ and
finally by the last complementarity condition, $x^*_t = 1$.
Therefore, we have $x_t(\bp^*) \le x_t^*$ for all $t$. On the other
hand, if $\pi_t < (\bp^*)^T\ba_t$, then we must have both
$x_t(\bp^*)$ and $x^*_t = 0$. Therefore, $x_t(\bp^*) = x_t^*$ if
$(\bp^*)^T \ba_t \neq \pi_t$. Under Assumption
\ref{assumption:ties}, there are at most $m$ values of $t$ such that
$(\bp^*)^T\ba_t = \pi_t$. Therefore, $x_t^*$ and $x_t(\bp^*)$
differs for no more than $m$ values of $t$. $\hfill\Box$

Lemma \ref{KKT} shows that if an optimal dual solution $\bp^*$ to
(\ref{primaloffline}) is known, then $x_t(\bp^*)$'s obtained by our
decision policy is close to the optimal offline solution. However,
in our online algorithm, we use the {\it sample dual price}
$\hat{\bp}$ learned from the first few inputs, which could be
different from the optimal dual price $\bp^*$. The remaining
discussion attempts to show that the sample dual price $\hat{\bp}$
will be sufficiently accurate for our purpose. In the following, we
will frequently use the fact that the random order assumption can be
interpreted as that the first $s$ inputs are uniform random samples
without replacement of size $s$ from the $n$ inputs. And we use $S$
to denote the sample set of size $s$, and $N$ to denote the complete
input set of size $n$. We start with the following lemma which shows
that with high probability, the primal solution $x_t(\hat{\bp})$
constructed using the sample dual price is feasible:

\begin{lemma}
\label{lem:one-time-constraints} The primal solution constructed
using the sample dual price is a feasible solution to the linear
program (\ref{primaloffline}) with high probability. More precisely,
with probability $1-\epsilon$,
$$\sum_{t=1}^n a_{it}x_t(\hat{\bp}) \le b_i, \ \ \forall i = 1, \ldots, m$$
given $B \ge \frac{6m\log (n/\epsilon)}{\epsilon^3}$.
\end{lemma}
{\noindent\bf Proof.} The proof will proceed as follows: Consider
any fixed price $\bp$ and $i$. We say a sample $S$ is ``bad'' for
this $\bp$ and $i$ if and only if $\bp$ is the optimal dual price to
(\ref{dualsample}) for the sample set $S$, but $\sum_{t=1}^{n}
a_{it}x_t(\bp) > b_i$. First, we show that the probability of bad
samples is small for every fixed $\bp$ and $i$. Then, we take a
union bound over all distinct prices to prove that with high
probability the learned price $\hat{\bp}$ will be such that
$\sum_{t=1}^{n}a_{it}x_t(\hat{\bp})\le b_i$ for all $i$.

To start with, we fix $\bp$ and $i$. Define $Y_t=a_{it}x_t(\V{p})$.
If $\V{p}$ is an optimal dual solution for the sample linear program
on $S$, applying Lemma \ref{KKT} to the sample problem, we have
\begin{eqnarray*}
\sum_{t\in S} Y_t=\sum_{t\in S} a_{it} x_t({\V{p}}) \le \sum_{t\in
S} a_{it}\tilde{x}_t\le  (1-\epsilon)\epsilon b_i,
\end{eqnarray*}
where $\tilde{\bx}$ is the primal optimal solution to the sample
linear program on $S$. Now we consider the probability of bad
samples for this $\bp$ and $i$:
\begin{align*}
&P\left(\sum_{t\in S} Y_t \le (1-\epsilon)\epsilon b_i, \sum_{t\in
N} Y_t \ge b_i\right).
\end{align*}
We first define $Z_t = \frac{b_iY_t}{\sum_{t\in N} Y_t}$. It is easy
to see that
\begin{align*}
&P\left(\sum_{t\in S} Y_t \le (1-\epsilon)\epsilon b_i, \sum_{t\in
N} Y_t \ge b_i\right) \le P\left(\sum_{t\in S} Z_t \le
(1-\epsilon)\epsilon b_i, \sum_{t\in N}Z_t = b_i\right).
\end{align*}
Furthermore, we have
\begin{align*}
P\left(\sum_{t\in S} Z_t \le (1-\epsilon)\epsilon b_i, \sum_{t\in
N}Z_t = b_i\right)  \le & \mbox{  }P\left(\left|\sum_{t\in S}
Z_t-\epsilon\sum_{t\in N} Z_{t}\right|\ge
\epsilon^2 b_i, \sum_{t\in N} Z_t=b_i\right)\\
\le & \mbox{  }P\left(\left|\sum_{t\in S} Z_t-\epsilon\sum_{t\in N}
Z_{t}\right|\ge
\epsilon^2 b_i \left| \sum_{t\in N} Z_t=b_i\right.\right)\\
\le & \mbox{
}2\exp{\left(\frac{-\epsilon^3b_i}{2+\epsilon}\right)}\le\delta
\end{align*}
where $\delta=  \frac{\epsilon}{m\cdot n^m}$. The second to last
step follows from the Hoeffding-Bernstein's Inequality for sampling
without replacement (Lemma \ref{HB} in Appendix \ref{app:one-time})
by treating $Z_t$, $t\in S$ as the samples without replacement from
$Z_t$, $t\in N$. We also used the fact that $0\le Z_t\le 1$ for all
$t$, therefore $\sum_{t\in N} (Z_t - \bar{Z})^2 \le \sum_{t\in N}
Z_t^2 \le b_i$ (and therefore the $\sigma_R^2$ in Lemma \ref{HB} can
be bounded by $b_i$). Finally, the last inequality is due to the
assumption made on $B$.

Next, we take a union bound over all {\it distinct}  $\V{p}$'s. We
call two price vectors $\bp$ and $\bq$ distinct if and only if they
result in distinct solutions, i.e., $\{x_t(\bp)\} \ne \{x_t(\bq)\}$.
Note that we only need to consider distinct prices, since otherwise
all the $Y_t$'s are exactly the same. Note that each distinct $\bp$
is characterized by a unique separation of $n$ points
($\{\pi_t,\V{a}_t\}_{t=1}^n$) in $m+1$-dimensional space by a
hyperplane. By results from computational geometry, the total number
of such distinct prices is at most $n^m$ (\cite{orlik}). Taking
union bound over the $n^m$ distinct prices, and $i=1,\ldots, m$, we
get the desired result. \hfill $\Box$\newline

Above we showed that with high probability, $x_t(\hat{\V{p}})$ is a
feasible solution. In the following, we show that it is also a
near-optimal solution.

\begin{lemma}
\label{lem:near-optimal-sample} The primal solution constructed
using the sample dual price is a near-optimal solution to the linear
program (\ref{primaloffline}) with high probability. More precisely,
with probability $1-\epsilon$,
\begin{equation*}
\sum_{t\in N} \pi_t x_t(\hat{\bp})\ge (1-3\epsilon) \OPT
\end{equation*}
given $B \ge \frac{6m\log (n/\epsilon)}{\epsilon^3}$.
\end{lemma}
{\noindent\bf Proof.} The proof is based on two observations. First,
$\{x_t(\hat{\bp})\}_{t=1}^n$ and $\hat{\bp}$ satisfy all the
complementarity conditions, and hence is the optimal primal and dual
solution to the following linear program:
\begin{equation}
\begin{array}{lll}\label{revisedprogram}
{\rm maximize} & \sum_{t\in N} \pi_t x_t & \\
\mbox{subject to} & \sum_{t\in N} a_{it}x_t \le \hat{b}_i, & i=1, \ldots, m\\
& 0 \le x_t \le 1, & t=1, \ldots, n
\end{array}
\end{equation}
where $\hat{b}_i = \sum_{t\in N} a_{it}x_t(\hat{\bp})$ if $\hat{p}_i
>0$, and $\hat{b}_i = \max\{\sum_{t\in N} a_{it}x_t(\hat{\bp}),b_i\}$, if $\hat{p}_i=0$.

Second, we show that if $\hat{p}_i>0$, then with probability
$1-\epsilon$, $\hat{b}_i \ge (1-3\epsilon)b_i$. To show this, let
$\hat{\bp}$ be the optimal dual solution of the sample linear
program on set $S$ and $\hat{\bx}$ be the optimal primal solution.
By the complementarity conditions of the linear program, if
$\hat{p}_i>0$, the $i^{th}$ constraint must be satisfied with
equality. That is, $\sum_{t\in S} a_{it} \hat{x}_t =
(1-\epsilon)\epsilon b_i$. Then, by Lemma \ref{KKT} and the
condition that $B=\min_i b_i \ge \frac{m}{\epsilon^2}$, we have
\begin{equation*}
\label{eq:inequality0} \sum_{t\in S} a_{it} x_t(\hat{\bp}) \ge
\sum_{t\in S} a_{it} \hat{x}_t -m\ge (1-2\epsilon)\epsilon b_i.
\end{equation*}
Then, using the Hoeffding-Bernstein's Inequality for sampling
without replacement, in a manner similar to the proof of Lemma
\ref{lem:one-time-constraints}, we can show that (the detailed proof
is given in Appendix \ref{app:deviation-inequality1}) given the
lower bound on $B$, with probability at least $1-\epsilon$, for all
$i$ such that $\hat{p}_i
> 0$:
\begin{equation}
\label{eq:deviation-inequality1} \hat{b}_i = \sum_{t\in N} a_{it}
x_t(\hat{\bp}) \ge (1-3\epsilon)b_i.
\end{equation}
Combined with the case $\hat{p}_i =0$, we know that with probability
$1-\epsilon$, $\hat{b}_i \ge (1-3\epsilon)b_i$ for all $i$. Lastly,
observing that whenever (\ref{eq:deviation-inequality1}) holds,
given an optimal solution $\bx^*$ to (\ref{primaloffline}),
$(1-3\epsilon)\bx^*$ will be a feasible solution to
(\ref{revisedprogram}). Therefore, the optimal value of
(\ref{revisedprogram}) is at least $(1-3\epsilon)\OPT$, which is
equivalently saying that
\begin{eqnarray*}
\sum_{t=1}^n\pi_tx_t(\hat{\V{p}})\ge(1-3\epsilon)\OPT.
\end{eqnarray*}\hfill$\Box$

Therefore, the objective value for the online solution taken over
the entire period is near-optimal. However, in the one-time learning
algorithm, no decision is made during the learning period $S$, and
only the decisions from periods $\{s+1, \ldots, n\}$ contribute to
the objective value. The following lemma that relates the optimal
value of the sample linear program (\ref{primalsample}) to the
optimal value of the offline linear program (\ref{primaloffline})
will bound the contribution from the learning period:
\begin{lemma}
\label{lem:learning-period-contribution} Let $\OPT(S)$ denote the
optimal value of the linear program (\ref{primalsample}) over sample
$S$, and $\OPT(N)$ denote the optimal value of the offline linear
program (\ref{primaloffline}) over $N$. Then,
$$\Ex[\OPT(S)] \le\epsilon \OPT(N).$$
\end{lemma}
{\noindent\bf Proof.} Let $(\V{x}^*, \bp^*, \V{y}^*)$ and
$(\hat{\V{x}}, \hat{\bp}, \hat{\V{y}})$ denote the optimal primal
and dual solutions of linear program (\ref{primaloffline}) on $N$,
and sample linear program (\ref{primalsample}) on $S$, respectively.
$$
\begin{array}{llll}
(\bp^*, \V{y}^*) = & \arg \min & \V{b}^T\bp + \sum_{t\in N} y_t & \\
& \mbox{s.t.} & \bp^T\V{a}_t + y_t \ge \pi_t, & \mbox{  } t\in N\\
& & \bp, \V{y} \ge 0 &
\end{array}
\
\begin{array}{llll}
(\hat{\bp}, \hat{\V{y}}) = & \arg \min & (1-\epsilon)\epsilon\V{b}^T\bp + \sum_{t\in S} y_t & \\
& \mbox{s.t.} & \bp^T\V{a}_t + y_t \ge \pi_t,\mbox{   }t\in S\\
& & \bp, \V{y} \ge 0. &
\end{array}
$$
Note that $S \subseteq N$, thus $(\bp^*, \V{y}^*)$ is a feasible
solution to the dual of the linear program on $S$. Therefore, by the
weak duality theorem:
\begin{eqnarray*}
\OPT(S) \le \epsilon\V{b}^T \bp^* + \sum_{t\in S} y^*_t.
\end{eqnarray*}
Therefore,
\begin{eqnarray*}
\Ex[\OPT(S)] \le \epsilon\V{b}^T \bp^* + \Ex\left[\sum_{t\in S}
y^*_t\right] = \epsilon(\V{b}^T \bp^* + \sum_{t\in N}y^*_t) =
\epsilon\OPT(N).\quad\quad\Box
\end{eqnarray*}

Now, we are ready to prove Proposition \ref{prop:one-time}:

{\bf\noindent Proof of Proposition \ref{prop:one-time}:} Using Lemma
\ref{lem:one-time-constraints} and Lemma
\ref{lem:near-optimal-sample}, with probability at least
$1-2\epsilon$, the following events happen:
$$\sum_{t=1}^n a_{it}x_t(\hat{\bp}) \le b_i, \ \ \ i = 1, \ldots, m$$
$$\sum_{t=1}^n \pi_t x_t(\hat{\bp})\ge (1-3\epsilon)OPT.$$
That is, the decisions $x_t(\hat{\bp})$ are feasible and the
objective value taken over the entire period $\{1, \ldots, n\}$ is
near-optimal. Denote this event by ${\cal E}$, where $P({\cal E})
\ge 1-2\epsilon$. We have by Lemma \ref{lem:one-time-constraints},
\ref{lem:near-optimal-sample} and
\ref{lem:learning-period-contribution}:
\begin{eqnarray*}
\Ex\left[\sum_{t=s+1}^n\pi_tx_t\right] & = &
\Ex\left[\sum_{t=1}^n\pi_tx_t -
\sum_{t=1}^s\pi_tx_t\right] \\
& \ge & \Ex\left[\sum_{t=1}^n\pi_tx_t(\hat{\V{p}}) I({\cal
E})\right] - \Ex\left[\sum_{t=1}^s \pi_tx_t(\hat{\V{p}})\right]\\
& \ge & (1-3\epsilon)P({\cal E})\OPT - \epsilon \OPT \ge
(1-6\epsilon)\OPT
\end{eqnarray*}
where $I(\cdot)$ is the indicator function, the first inequality is
because under ${\cal E}$, $x_t = x_t(\hat{\bp})$, and the second
last inequality uses the fact that $x_t(\hat{\bp}) \le \hat{x}_t$
which is due to Lemma \ref{KKT}. $\hfill\Box$

%

\section{Dynamic Learning Algorithm}
\label{sec:dynamic}

The algorithm discussed in Section \ref{sec:one-time} uses the first
$\epsilon n$ inputs to learn a threshold price, and then applies it
in the remaining time horizon. While this algorithm has its own
merits, in particular, requires solving only a small linear program
defined on $\epsilon n$ variables, the lower bound required on $B$
is stronger than that claimed in Theorem \ref{th:main} by an
$\epsilon$ factor.

In this section, we propose an improved dynamic learning algorithm
that will achieve the result in Theorem \ref{th:main}. Instead of
computing the price only once, the dynamic learning algorithm will
update the price every time the history doubles, that is, it learns
a new price at time $t = \epsilon n, 2\epsilon n, 4\epsilon n,
\ldots$. To be precise, let $\hat{\V{p}}^{\ell}$ denote the optimal
dual solution for the following partial linear program defined on
the inputs until time $\ell$:
\begin{equation}
\begin{array}{lll}\label{lstep}
\mbox{maximize} &\sum_{t=1}^{\ell} \pi_t x_t & \\
\mbox{subject to}& \sum_{t=1}^{\ell} a_{it}x_t\le
(1-h_{\ell})\frac{\ell}{n} b_i, & i=1,\ldots, m\\
& 0 \le x_t \le 1, & t=1, \ldots, \ell
\end{array}
\end{equation}
where the set of numbers $h_{\ell}$ are defined as follows:
\begin{equation*}
\begin{array}{l}\label{definitionofh}
h_{\ell}=\epsilon\sqrt{\frac{n}{\ell}}.
\end{array}
\end{equation*}
Also, for any given dual price vector $\V{p}$, we define the same
allocation rule $x_t(\bp)$ as in (\ref{linear:xp}). Our dynamic
learning algorithm is stated as follows:\\

{\bf\noindent Algorithm DLA (Dynamic Learning Algorithm):}
\begin{enumerate}
\item Initialize $t_0= \epsilon n$. Set $x_t = 0$, for all
$t\le t_0$.
\item Repeat for $t=t_0+1, t_0+2, \ldots$
\begin{enumerate}
\item Set $\hat{x}_t = x_t(\hat{\bp}^{\ell})$. Here $\ell= 2^r\epsilon n$ where $r$ is the largest integer such that $\ell<t$.
\item If $a_{it}\hat{x}_t\le b_i-\sum_{j=1}^{t-1}a_{ij}x_j$ for all $i$,
then set $x_t=\hat{x}_t$; otherwise, set $x_t=0$. Output $x_t$.\\
\end{enumerate}
\end{enumerate}

Note that we update the dual price vector $\lceil
\log_{2}{(1/\epsilon)}\rceil$ times during the entire time horizon.
Thus, the dynamic learning algorithm requires more computation.
However, as we show next, it requires a weaker lower bound on $B$
for proving the same competitive ratio. The intuition behind this
improvement is as follows. Note that initially, at $\ell=\epsilon
n$, $h_{\ell} = \sqrt{\epsilon} > \epsilon$. Thus, we have larger
slacks at the beginning, and the large deviation argument for
constraint satisfaction (as in Lemma \ref{lem:one-time-constraints})
requires a weaker condition on $B$. As $t$ increases, $\ell$
increases, and $h_{\ell}$ decreases. However, for larger values of
$\ell$, the sample size is larger, making a weaker condition on $B$
sufficient to prove the same error bound. Furthermore, $h_{\ell}$
decreases rapidly enough, such that the overall loss on the
objective value is not significant. As one will see, the careful
choice of the numbers $h_{\ell}$ plays an important role in proving
our results.

\subsection{Competitive Ratio Analysis}
\label{sec:competitive_ratio_analysis2}

The analysis for the dynamic learning algorithm proceeds in a manner
similar to that for the one-time learning algorithm. However,
stronger results for the price learned in each period need to be
proved here. In the following, we assume $\epsilon = 2^{-E}$ and let
$L=\{\epsilon n, 2\epsilon n, \ldots, 2^{E-1}\epsilon n\}$.

Lemma \ref{prop:constraints} and \ref{prop:objective1} are parallel
to Lemma \ref{lem:one-time-constraints} and
\ref{lem:near-optimal-sample} in Section \ref{sec:one-time}, however
require a weaker condition on $B$:
\begin{lemma}
\label{prop:constraints} For any $\epsilon>0$, with probability
$1-\epsilon$:
$$\sum_{t=\ell+1}^{2\ell} a_{it} x_t(\hat{\bp}^{\ell})\le \frac{\ell}{n}b_i, \ \ \ \ \mbox{ for all } i\in \{1, \ldots, m\}, \ \ell \in L$$
given $B=\min_i b_i \ge \frac{10m\log{(n/\epsilon)}}{\epsilon^2}$.
\end{lemma}
{\noindent\bf Proof.} The proof is similar to the proof of Lemma
\ref{lem:one-time-constraints} but a more careful analysis is
needed. We provide a brief outline here with a detailed proof in
Appendix \ref{proof: dynamic}. First, we fix $\V{p}$, $i$ and
$\ell$. This time, we say a permutation is ``bad'' for this $\V{p}$,
$i$ and $\ell$ if and only if $\V{p}=\hat{\V{p}}^l$ (i.e., $\V{p}$
is the learned price under the current arrival order) but
$\sum_{t=\ell+1}^{2\ell}a_{it}x_t(\hat{\V{p}}^l)>\frac{l}{n}b_i$. By
using the Hoeffding-Bernstein's Inequality for sampling without
replacement, we show that the probability of ``bad'' permutations is
less than $\delta=\frac{\epsilon}{m\cdot n^m\cdot E}$ for any fixed
$\V{p}$, $i$ and $\ell$ under the condition on $B$. Then by taking a
union bound over all distinct prices, all items $i$ and periods
$\ell$, the lemma is proved. \hfill $\Box$\newline

In the following, we use $\LP_{s}(\V{d})$ to denote the partial
linear program that is defined on variables till time $s$ with
right-hand-side in the inequality constraints set as $\V{d}$. That
is,
$$ \LP_s(\V{d}): \ \ \ \begin{array}{lll}
{\rm maximize} & \sum_{t=1}^{s} \pi_t x_t &\\
\mbox{subject to} & \sum_{t=1}^{s} a_{it} x_t \le d_i, & i=1, \ldots, m\\
& 0\le x_t \le 1, & t=1, \ldots, s.
\end{array}
$$
And let $\OPT_s(\V{d})$ denote the optimal objective value for
$\LP_s(\V{d})$.
\begin{lemma}
\label{prop:objective1} With probability at least $1-\epsilon$, for
all $\ell \in L$:
$$\sum_{t=1}^{2\ell} \pi_t x_t(\hat{\bp}^{\ell}) \ge (1-2h_{\ell}-\epsilon)\OPT_{2\ell}\left(\frac{2\ell}{n}\V{b}\right)$$
given $B=\min_i b_i \ge \frac{10m\log{(n/\epsilon)}}{\epsilon^2}$.
\end{lemma}
{\noindent\bf Proof.} Let $\hat{b}_i = \sum_{j=1}^{2\ell} a_{ij}
x_j(\hat{\bp}^{\ell})$ for $i$ such that $\hat{p}^{\ell}_i > 0$, and
$\hat{b}_i = \max\{\sum_{j=1}^{2\ell} a_{ij} x_j(\hat{\bp}^{\ell}),
\frac{2\ell}{n}b_i\}$, otherwise. Then the solution pair
$(\{x_t(\hat{\bp}^{\ell})\}_{t=1}^{2\ell}, \hat{\bp}^{\ell})$
satisfy all the complementarity conditions, thus are optimal
solutions (primal and dual respectively) to the linear program
$\LP_{2\ell}(\hat{\V{b}})$:
\begin{equation*}
\begin{array}{lll}
{\rm maximize} & \sum_{t=1}^{2\ell} \pi_t x_t &\\
\mbox{subject to} & \sum_{t=1}^{2\ell} a_{it} x_t \le \hat{b}_i, & i=1, \ldots, m\\
& 0\le x_t \le 1, & t=1, \ldots, 2\ell.
\end{array}
\end{equation*}
This means
\begin{equation*}
\sum_{t=1}^{2\ell} \pi_t x_t(\hat{\bp}^{\ell}) =
\OPT_{2\ell}(\V{\hat{b}}) \ge \left(\min_i \frac{\hat{b}_i}{b_i
\frac{2\ell}{n}}
\right)\OPT_{2\ell}\left(\frac{2\ell}{n}\V{b}\right).
\end{equation*}
Now, we analyze the ratio $\frac{\hat{b}_i}{2\ell b_i/n}$. By
definition, for $i$ such that $\hat{p}^{\ell}_i=0$, $\hat{b}_i \ge
2\ell b_i/n$. Otherwise, using techniques similar to the proof of
Lemma \ref{prop:constraints}, we can prove that with probability
$1-\epsilon$, for all $i$,
\begin{equation}
\label{eq:deviation-inequality2} \hat{b}_i = \sum_{t=1}^{2\ell}
a_{it} x_t(\hat{\bp}^{\ell})\ge
(1-2h_{\ell}-\epsilon)\frac{2\ell}{n}b_i.
\end{equation}
A detailed proof of (\ref{eq:deviation-inequality2}) appears in
Appendix \ref{app:deviation-inequality2}. And the lemma follows from
(\ref{eq:deviation-inequality2}). \hfill $\Box$\newline

Next, similar to Lemma \ref{lem:learning-period-contribution} in the
previous section, we prove the following lemma relating the optimal
value of the sample linear program to the optimal value of the
offline linear program:
\begin{lemma}
\label{lem:learning-period-contribution2} For any $\ell$,
$$\Ex\left[\OPT_\ell\left(\frac{\ell}{n}\V{b}\right)\right] \le \frac{\ell}{n} \OPT.$$
\end{lemma}
The proof of lemma \ref{lem:learning-period-contribution2} is
exactly the same as the proof for Lemma
\ref{lem:learning-period-contribution} thus we omit its proof.

Now we are ready to prove Theorem \ref{th:main}.

{\noindent\bf Proof of Theorem \ref{th:main}:} Observe that the
output of the online solution at time $t\in \{\ell +1, \ldots,
2\ell\}$ is $x_t(\hat{\bp}^{\ell})$ as long as the constraints are
not violated. By Lemma \ref{prop:constraints} and Lemma
\ref{prop:objective1}, with probability at least $1-2\epsilon$:
$$\sum_{t=\ell+1}^{2\ell} a_{it} x_t(\hat{\bp}^{\ell})\le \frac{\ell}{n}b_i, \ \ \ \ \mbox{ for all } i\in \{1, \ldots, m\}, \ \ell \in L$$
$$\sum_{t=1}^{2\ell} \pi_t x_t(\hat{\bp}^{\ell}) \ge (1-2h_{\ell}-\epsilon)\OPT_{2\ell}\left(\frac{2\ell}{n}\V{b}\right), \quad \mbox{for all }\ell \in L.$$
Denote this event by ${\cal E}$, where $P({\cal E}) \ge
1-2\epsilon$. The expected objective value achieved by the online
algorithm can be bounded as follows:
\begin{eqnarray*}
& & \Ex\left[\sum_{\ell \in L}\sum_{t=\ell+1}^{2\ell} \pi_t
x_t\right]\\
&\ge  & \Ex\left[\sum_{\ell \in L}\sum_{t=\ell+1}^{2\ell} \pi_t x_t(\hat{\bp}^{\ell})I({\cal E})\right] \\
&\ge&\sum_{l\in
L}\Ex\left[\sum_{t=1}^{2\ell}\pi_tx_t(\hat{\V{p}}^l)I(\cal E)
\right]-\sum_{\ell\in L}\Ex\left[\sum_{t=1}^\ell
\pi_tx_t(\hat{\V{p}}^\ell)I(\cal E) \right]\\
& \ge& \sum_{\ell\in
L}(1-2h_l-\epsilon)\Ex\left[\OPT_{2\ell}\left(\frac{2\ell}{n}\V{b}\right)I(\cal
E)\right]-\sum_{\ell\in L}
\Ex\left[\OPT_{\ell}\left(\frac{\ell}{n}\V{b}\right)I(\cal E)\right]\\
& \ge&P(\cal E) \cdot\OPT-\sum_{\ell\in
L}2h_l\Ex\left[\OPT_{2\ell}\left(\frac{2\ell}{n}\V{b}\right) I(\cal
E)\right]-\epsilon\sum_{\ell\in L}
\Ex\left[\OPT_{2\ell}\left(\frac{2\ell}{n}\V{b}\right)I(\cal E)\right]-\Ex\left[\OPT_{\epsilon n}(\epsilon\V{b})I(\cal E)\right]\\
& \ge& (1-2\epsilon)\OPT-\sum_{\ell\in
L}2h_l\Ex\left[\OPT_{2\ell}\left(\frac{2\ell}{n}\V{b}\right)\right]-\epsilon\sum_{\ell\in
L}\Ex\left[\OPT_{2\ell}\left(\frac{2\ell}{n}\V{b}\right)\right]-\Ex\left[\OPT_{\epsilon
n}(\epsilon \V{b})\right]\\
& \ge& (1-2\epsilon)\OPT- 4\sum_{\ell\in
L}\frac{h_{\ell}\ell}{n}\OPT-2\epsilon\sum_{l\in
L}\frac{\ell}{n}\OPT-\epsilon\OPT\\
& \ge & (1-15\epsilon)\OPT.
\end{eqnarray*}
The third inequality is due to Lemma \ref{prop:objective1}, the
second to last inequality is due to Lemma
\ref{lem:learning-period-contribution2} and the last inequality
follows from the fact that
$$\sum_{\ell \in L }\frac{\ell}{n} = (1-\epsilon), \mbox{ and }
\sum_{\ell \in L} h_{\ell} \frac{\ell}{n} = \epsilon \sum_{\ell\in
L} \sqrt{\frac{\ell}{n}} \le 2.5\epsilon.$$ Therefore, Theorem
\ref{th:main} is proved. \hfill $\Box$

\section{Worst-case Bound for any Algorithm}
\label{sec:low}

In this section, we prove Theorem \ref{th:lowerbound}, i.e., the
condition $B\ge \Omega(\log m/\epsilon^2)$ is necessary for any
online algorithm to achieve a competitive ratio of $1-O(\epsilon)$.
We prove this by constructing an instance of (\ref{primaloffline})
with $m$ items and $B$ units of each item such that no online
algorithm can achieve a competitive ratio of $1-O(\epsilon)$ unless
$B\ge \Omega(\log m/\epsilon^2)$.

In this construction, we refer to the $0-1$ vectors $\V{a}_t$'s as
demand vectors, and $\pi_t$'s as profit coefficients. Assume
$m=2^{z}$ for some integer $z$. We will construct $z$ pairs of
demand vectors such that the demand vectors in each pair are
complement to each other, and do not share any item.

However, every set of $z$ vectors consisting of exactly one vector
from each pair will share at least one common item. To achieve this,
consider the $2^{z}$ possible boolean strings of length $z$. The
$j^{th}$ boolean string represents $j^{th}$ item for $j=1,\ldots,
m=2^z$ (for illustrative purpose, we index the item from 0 in our
later discussion). Let $s_{ij}$ denote the value at $i^{th}$ bit of
the $j^{th}$ string. Then, we construct a pair of demand vectors
$\V{v}_i, \V{w}_i \in \{0,1\}^m$, by setting $v_{ij}=s_{ij}$,
$w_{ij} = 1-s_{ij}$.

Table \ref{table:counterexample} illustrates this construction for
$m=8$ ($z =3$):


\begin{table}[th]
\begin{center}
\begin{tabular}{p{0.5cm}p{0.7cm}p{0.7cm}p{0.7cm}p{0.7cm}p{1.7cm}p{0.5cm}p{0.7cm}p{0.7cm}p{0.7cm}p{0.7cm}}
 \multicolumn{2}{c}{} & \multicolumn{3}{c}{Demand vectors} &
\multicolumn{3}{c}{} & \multicolumn{3}{c}{Demand vectors} \\
\multicolumn{2}{c}{} & $\mbox{ }\bf{v}_3$ & $\mbox{ }\bf{v}_2$ &
$\mbox{ }\bf{v}_1$ &
\multicolumn{3}{c}{} & $ \mbox{ }\bf{w}_3$ & $\mbox{ }\bf{w}_2$ & $\mbox{ }\bf{w}_1$ \\
\cline{3-5} \cline{9-11} \multirow{8}{*}{\begin{turn}{90}
Items\end{turn}} & 0 & \multicolumn{1}{|c|}0 & \multicolumn{1}{c|}0
& \multicolumn{1}{c|}0 & &
\multirow{8}{*}{\begin{turn}{90} Items\end{turn}} & 0 & \multicolumn{1}{|c|}1 & \multicolumn{1}{c|}1 & \multicolumn{1}{c|}1 \\
& 1 & \multicolumn{1}{|c|} 0 & \multicolumn{1}{c|}0 & \multicolumn{1}{c|}1 & & & 1 & \multicolumn{1}{|c|}1 & \multicolumn{1}{c|}1 & \multicolumn{1}{c|}0\\
& 2 & \multicolumn{1}{|c|} 0 & \multicolumn{1}{c|}1 & \multicolumn{1}{c|}0 & & & 2 & \multicolumn{1}{|c|}1 & \multicolumn{1}{c|}0 & \multicolumn{1}{c|}1\\
& 3 & \multicolumn{1}{|c|} 0 & \multicolumn{1}{c|}1 & \multicolumn{1}{c|}1 & & & 3 & \multicolumn{1}{|c|}1 & \multicolumn{1}{c|}0 & \multicolumn{1}{c|}0\\
& 4 & \multicolumn{1}{|c|} 1 & \multicolumn{1}{c|}0 & \multicolumn{1}{c|}0 & & & 4 & \multicolumn{1}{|c|}0 & \multicolumn{1}{c|}1 & \multicolumn{1}{c|}1\\
& 5 & \multicolumn{1}{|c|} 1 & \multicolumn{1}{c|}0 & \multicolumn{1}{c|}1 & & & 5 & \multicolumn{1}{|c|}0 & \multicolumn{1}{c|}1 & \multicolumn{1}{c|}0\\
& 6 & \multicolumn{1}{|c|} 1 & \multicolumn{1}{c|}1 &
\multicolumn{1}{c|}0 & & & 6 & \multicolumn{1}{|c|}0 &
\multicolumn{1}{c|}0 & \multicolumn{1}{c|}1\\ & 7 &
\multicolumn{1}{|c|} 1 & \multicolumn{1}{c|}1 & \multicolumn{1}{c|}1
& & & 7 & \multicolumn{1}{|c|}0 & \multicolumn{1}{c|}0 &
\multicolumn{1}{c|}0 \\ \cline{3-5} \cline{9-11}
\end{tabular}\caption{Illustration of the worst-case
bound}\label{table:counterexample}
\end{center}
\end{table}
Note that the pair of vectors $\V{v}_i, \V{w_i}, i=1,\ldots, z$ are
complement to each other. Consider any set of $z$ demand vectors
formed by picking exactly one of the two vectors $\V{v}_i$ and
$\V{w}_i$ for each $i=1,\ldots, z$. Then form a bit string by
setting $s_{ij}=1$ if this set has vector $\V{v}_i$ and $0$ if it
has vector $\V{w}_i$. Then, all the vectors in this set share the
item corresponding to the boolean string. For example, in Table
\ref{table:counterexample}, the demand vectors $\V{v}_3, \V{w}_2,
\V{w}_1$ share item $4 (='100')$, the demand vectors $\V{w}_3,
\V{v}_2, \V{v}_1$ share item $3(='011')$ and so on.

Now, we construct an instance consisting of
\begin{itemize}
\item $B/z$ inputs with profit coefficient $4$ and demand vector $\V{v}_i$, for each $i=1,\ldots, z$.
\item $q_i$ inputs with profit $3$ and demand vector $\V{w_i}$, for each $i$, where $q_i$ is a random variable following Binomial($2B/z,1/2$).
\item $\sqrt{B/4z}$ inputs with profit $2$ and demand  vector $\V{w_i}$, for each $i$.
\item $2B/z-q_i$ inputs with profit $1$ and demand vector $\V{w}_i$, for each $i$.
\end{itemize}

Using the properties of demand vectors ensured in the construction,
we prove the following claim:

\begin{claim}
\label{claim:low} Let $r_i$ denote the number of vectors of type
$\V{w}_i$ accepted by any $1-\epsilon$ competitive solution for the
constructed example.  Then, it must hold that
$$ \sum_i |r_i - B/z| \le 7\epsilon B.$$
\end{claim}
{\noindent\bf Proof.} Let $\OPT$ denote the optimal value of the
offline problem. And let $\OPT_i$ denote the profit obtained from
demands accepted of type $i$. Let $\topw_i(k)$ denote the sum of
profits of top $k$ inputs with demand vector $\V{w}_i$. Then
\vspace{-0.1in}
$$\OPT = \sum_{i=1}^{z} \OPT_i \ge \sum_{i=1}^{z} (4B/z+\topw_i(B/z)) = 4B + \sum_{i=1}^{z} \topw_i(B/z) \vspace{-0.1in}.$$
Let $\OBJ$ be the objective value of a solution which accepts $r_i$
vectors of type $\V{w}_i$. First, note that $\sum_i r_i \le B$. This
is because all $\V{w}_i$s share one common item, and there are at
most $B$ units of this item available. Let $Y$ be the set $\{i:r_i >
B/z\}$, and $X$ be the remaining $i$'s, i.e. $X=\{i:r_i\le B/z\}$.
Then, we show that the total number of accepted $\V{v}_i$s cannot be
more than $B-\sum_{i\in Y} r_i + |Y| B/z$. Obviously, the set $Y$
cannot contribute more than $|Y|B/z$ $\V{v}_i$s. Let $S\subseteq X$
contribute the remaining $v_i$s. Now consider the item that is
common to all $\V{w}_i$s in set $Y$ and $\V{v}_i$s in the set $S$
(there is at least one such item by construction). Since only $B$
units of this item are available, the total number of $v_i$s
contributed by $S$ cannot be more than $B-\sum_{i\in Y} r_i$.
Therefore the number of accepted $\V{v}_i$s is less than or equal to
$B-\sum_{i\in Y} r_i+|Y|B/z$.

Denote $P=\sum_{i\in Y} r_i - |Y| B/z$, $M=|X| B/z - \sum_{i\in X}
r_i$. Then, $P, M \ge 0$. And the objective value
\begin{eqnarray*}
\OBJ & \le & \sum_{i=1}^{z} \topw_i(r_i)+ 4(B-\sum_{i\in Y} r_i + |Y|B/z) \\
& \le &\sum_{i=1}^z \topw_i(B/z) + 3P-M+ 4(B-P) \\
& = & \OPT -P-M.
\end{eqnarray*}

Since $\OPT \le 7B$, this means that, $P+M$ must be less than
$7\epsilon B$ in order to get an approximation ratio of $1-\epsilon$
or better. \hfill $\Box$

Here is a brief description of the remaining proof. By construction,
for every $i$, there are exactly $2B/z$ demand vectors $\V{w}_i$
that have profit coefficients $1$ and $3$, and among them each has
equal probability to take value $1$ or $3$. Now, from the previous
claim, in order to get a near-optimal solution, one must select
close to $B/z$ demand vectors of type $\V{w}_i$. Therefore, if the
total number of $(3, \V{w}_i)$ inputs are more than $B/z$, then
selecting any $(2,\V{w}_i)$ will cause a loss of $1$ in profit as
compared to the optimal profit; and if the total number of $(3,
\V{w}_i)$ inputs are less than $B/z-\sqrt{B/4z}$, then rejecting any
$(2,\V{w}_i)$ will cause a loss of $1$ in profit. Using the central
limit theorem, at any step, both these events can happen with a
constant probability. Thus, every decision for $(2,\V{w}_i)$ might
result in a loss with constant probability, which results in a total
expected loss of $\Omega(\sqrt{B/z})$ for every $i$, that is, a
total loss of $\Omega(\sqrt{z B})$.

If the number of $\V{w}_i$s to be accepted is not exactly $B/z$,
some of these $\sqrt{B/z}$ decisions may not be mistakes, but as in
the claim above, such cases cannot be more than $7\epsilon B$.
Therefore, the expected value of online solution,
$$ \ONLINE \le \OPT - \Omega(\sqrt{z B}-7\epsilon B).$$
Since $\OPT \le 7B$, in order to get $(1-\epsilon)$ approximation
factor, we need
$$ \Omega(\sqrt{z/B} - 7\epsilon) \le 7\epsilon \Rightarrow  B \ge \Omega(z/\epsilon^2) = \Omega(\log(m)/\epsilon^2).$$

This completes the proof of Theorem \ref{th:lowerbound}. A detailed
exposition of the steps used in this proof appears in  Appendix
\ref{app:lower-bound}.

\section{Extensions}
\label{sec:ext} We provide a few extensions of our results in this
section.

\subsection{Online multi-dimensional linear program}
\label{sec:ext-convex} We consider the following more general online
linear programs with multi-dimensional decisions $\bx_t \in
\mathbb{R}^k$ at each step, as defined in
(\ref{multidimensionalonline}) in Section \ref{sec:intro}:
\begin{equation}
\begin{array}{lll}
\label{convexoffline}
{\rm maximize} & \sum_{t=1}^n \V{f}^T_t \bx_t & \\
\mbox{subject to} & \sum_{t=1}^n \V{g}^T_{it}\bx_t \le b_i, & i=1, \ldots, m\\
& \bx_t^T\V{e} \le 1, \bx_t \ge 0, & t=1,\ldots,n\\
& \bx_t \in \mathbb{R}^{k}, & t = 1,\ldots, n.
\end{array}
\end{equation}
Our online algorithm remains essentially the same (as described in
Section \ref{sec:dynamic}), with $\bx_t(\bp)$ now defined as
follows:
\begin{equation*}
{\bx}_t(\bp) = \left\{
\begin{array}{ll}\label{outputgeneral}
0 & \mbox{ if  for all $j$, } f_{tj}  \le \sum_i p_ig_{itj} \\
\V{e}_r & \mbox{ otherwise, where } r \in \arg \max_{j} (f_{tj} -
\sum_i p_ig_{itj} ).
\end{array}\right.
\end{equation*}
Here $\V{e}_r$ is the unit vector with $1$ at the $r$th entry and
$0$ otherwise. And we break ties arbitrarily in our algorithm. Using
the complementarity conditions of (\ref{convexoffline}), and the
lower bound condition on $B$ as assumed in Theorem \ref{th:convex},
we can prove the following lemmas.\endnote{Here we make an
assumption similar to Assumption \ref{assumption:ties}. That is, for
any $\bp$, there can be at most $m$ arrivals such that there are
ties in $f_{tj} - \sum_{i}p_ig_{itj}$. As argued in the discussions
following Assumption \ref{assumption:ties}, this assumption is
without loss of generality.} The proofs are very similar to the
proofs for the one-dimensional case, and will be provided in
Appendix \ref{app:multi_dimension}.
\begin{lemma}\label{multilemma1}
Let $\bx^*$ and $\bp^*$ be the optimal primal and dual solutions to
(\ref{convexoffline}) respectively. Then $\bx^*_t$ and
$\bx_t(\bp^*)$ differs for at most $m$ values of $t$.
\end{lemma}
\begin{lemma}\label{multilemma2}
Define $\bp$ and $\bq$ to be {\it distinct} if and only if
$\V{x}_t(\bp) \ne \V{x}_t(\bq)$ for some $t$. Then, there are at
most $n^mk^{2m}$ distinct price vectors.
\end{lemma}
\noindent With the above lemmas, the proof of Theorem
\ref{th:convex} will follow exactly as the proof for Theorem
\ref{th:main}.

\subsection{Online integer programs}
\label{sec:ext-integer} From the definition of $x_t(\bp)$ in
(\ref{linear:xp}), our algorithm always outputs integer solutions.
And since the competitive ratio analysis compares the online
solution to the optimal solution of the corresponding linear
programming relaxation, the competitive ratio stated in Theorem
\ref{th:main} also holds for the online integer programs. The same
observation holds for the general online linear programs introduced
in Section \ref{sec:ext-convex} since it also outputs integer
solutions.

\subsection{Fast solution for large linear programs by column sampling}
Apart from online problems, our algorithm can also be applied for
solving (offline) linear programs that are too large to consider all
the variables explicitly. Similar to the one-time learning online
solution, one could randomly sample of $\epsilon n$ variables, and
use the dual solution $\hat{\bp}$ for this smaller program to set
the values of variables $x_j$
 as $x_j(\hat{\bp})$.
This approach is very similar to the column generation method used
for solving large linear programs \cite{dantzig}. Our result
provides the first rigorous analysis of the approximation achieved
by the approach of reducing the linear program size by randomly
selecting a subset of columns.

\section{Conclusions}
\label{sec:conclusions} In this paper, we provide a $1-O(\epsilon)$
competitive algorithm for a general class of online linear
programming problems under the assumption of random order of arrival
and some mild conditions on the right-hand-side input. The
conditions we use are independent of the optimal objective value,
the objective coefficients, and the distributions of input data.

Our dynamic learning algorithm works by dynamically updating a
threshold price vector at geometric time intervals, where the dual
prices learned from the revealed columns in the previous period are
used to determine the sequential decisions in the current period.
Our dynamic learning approach might be useful in designing online
algorithms for other problems.

There are many questions for future research. One important question
is whether the current bound on the size of the right-hand-input $B$
is  tight? Currently as we show in this paper, there is a gap
between our algorithm and the lower bound. Through some numerical
experiments, we find that the actual performance of our algorithm is
close to the lower bound (see \cite{wang_thesis}). However, we are
not able to prove it. Filling that gap would be a very interesting
direction for future research.

%
%
\begin{APPENDICES}

\section{Supporting lemmas for Section \ref{sec:one-time}}
\label{app:one-time}

\subsection{Hoeffding-Bernstein's Inequality for sampling without replacement}
By Theorem 2.14.19 in \cite{vandervaat}:
\begin{lemma}
\label{HB} Let $u_1,u_2,...u_r$ be random samples without
replacement from the real numbers $\{c_1,c_2,...,c_R\}$. Then for
every $t>0$,
\begin{eqnarray*}
P\left(\left|\sum_{i=1}^r u_i-r\bar{c}\right|\ge t\right)\le
2\exp\left(-\frac{t^2}{2r\sigma_R^2+t\Delta_R}\right)
\end{eqnarray*}
where $\Delta_R=\max_{i}c_i-\min_{i}c_i$,  $\bar{c} =
\frac{1}{R}\sum_i c_i$, and
$\sigma_R^2=\frac{1}{R}\sum_{i=1}^R(c_i-\bar{c})^2$.
\end{lemma}

\subsection{Proof of inequality (\ref{eq:deviation-inequality1})}
\label{app:deviation-inequality1} We prove that with probability
$1-\epsilon$, $\hat{b}_i = \sum_{t\in N} a_{it} x_t(\hat{\bp}) \ge
(1-3\epsilon)b_i$ given $\sum_{t\in S} a_{it} x_t(\hat{\bp}) \ge
(1-2\epsilon)\epsilon b_i$. The proof is very similar to the proof
of Lemma \ref{lem:one-time-constraints}. Fix a price vector $\bp$
and $i$. Define a permutation is ``bad'' for $\bp, i$ if both (a)
$\sum_{t\in S} a_{it} x_t(\bp) \ge (1-2\epsilon)\epsilon b_i$ and
(b) $\sum_{t\in N} a_{it} x_t(\bp) \le (1-3\epsilon)b_i$ hold.

Define $Y_t = a_{it}x_t(\bp)$. Then, the probability of bad
permutations is bounded by:
\begin{eqnarray*}
P\left(\left|\sum_{t\in  S} Y_t - \epsilon \sum_{t\in N} Y_t\right|
\ge \epsilon ^2b_i\left|\sum_{t\in N}Y_t\le
(1-3\epsilon)b_i\right.\right) & \le & P\left(\left|\sum_{t\in  S}
Z_t - \epsilon \sum_{t\in N} Z_t\right| \ge \epsilon
^2b_i\left|\sum_{t\in N}Z_t =
(1-3\epsilon)b_i\right.\right) \\
& \le &  2\exp\left(-\frac{b_i\epsilon^3}{3}\right) \le
\frac{\epsilon}{m \cdot n^m}
\end{eqnarray*}
where $Z_t = \frac{(1-3\epsilon)b_i Y_t}{\sum_{t\in N}Y_t}$ in the
first inequality and the second inequality is because of Lemma
\ref{HB} and the last inequality follows from that $b_i \ge
\frac{6m\log(n/\epsilon)}{\epsilon^3}$. Summing over $n^m$ distinct
prices and $i=1, \ldots, m$, we get the desired inequality. \hfill
$\Box$

\section{Supporting lemmas for Section \ref{sec:dynamic}}

\subsection{Proof of Lemma \ref{prop:constraints}}
\label{proof: dynamic} Consider $\sum_{t} a_{it}\hat{x}_t$ for a
fixed $i$. For ease of notation, we temporarily omit the subscript
$i$. Define $Y_t=a_tx_t(\V{p})$.
 If $\V{x}$ and $\V{p}$ are the optimal primal and
dual solutions for (\ref{lstep}) and its dual respectively, then we
have:
\begin{eqnarray*}
\sum_{t=1}^{\ell} Y_t=\sum_{t=1}^{\ell} a_t x_t({\V{p}}) \le
\sum_{t=1}^{\ell} a_t x_t \le (1-h_{\ell})b\frac{\ell}{n}.
\end{eqnarray*}
Here the first inequality is because of the definition of
$x_t(\V{p})$ and Lemma \ref{KKT}. Therefore, the probability of
``bad'' permutations for this $\V{p}$, $i$ and $\ell$ is bounded by:
\begin{eqnarray}\label{twoterms}
& & P\left(\sum_{t=1}^{\ell} Y_t \le (1-h_{\ell})\frac{b\ell}{n},\ \sum_{t=\ell+1}^{2\ell} Y_t \ge \frac{b\ell}{n}\right) \nonumber\\
& \le & P\left(\sum_{t=1}^{\ell} Y_t \le
(1-h_{\ell})\frac{b\ell}{n}, \sum_{t=1}^{2\ell} Y_t \ge
\frac{2b\ell}{n}\right)+ P\left(\left|\sum_{t=1}^{\ell} Y_t -
\frac{1}{2}\sum_{t=1}^{2\ell} Y_t\right| \ge \frac{h_\ell}{2}
\frac{b\ell}{n}, \ \sum_{t=1}^{2\ell} Y_t \le
\frac{2b\ell}{n}\right).
\end{eqnarray}
For the first term, we first define $Z_t = \frac{2b\ell
Y_t}{n\sum_{t=1}^{2\ell} Y_t}$. It is easy to see that
\begin{eqnarray*}
P\left(\sum_{t=1}^{\ell} Y_t \le (1-h_{\ell})\frac{b\ell}{n},
\sum_{t=1}^{2\ell} Y_t \ge \frac{2b\ell}{n}\right) \le
P\left(\sum_{t=1}^{\ell} Z_t \le (1-h_{\ell})\frac{b\ell}{n},
\sum_{t=1}^{2\ell} Z_t = \frac{2b\ell}{n}\right).
\end{eqnarray*}
And furthermore, using Lemma \ref{HB}, we have
\begin{eqnarray*}
P\left(\sum_{t=1}^{\ell} Z_t \le (1-h_{\ell})\frac{b\ell}{n},
\sum_{t=1}^{2\ell} Z_t = \frac{2b\ell}{n}\right)& \le &
P\left(\sum_{t=1}^{\ell} Z_t \le (1-h_{\ell})\frac{b\ell}{n} \left|
\sum_{t=1}^{2\ell} Z_t = \frac{2b\ell}{n}\right.\right)\\
& \le &  P\left(\left|\sum_{t=1}^{\ell} Z_t -
\frac{1}{2}\sum_{t=1}^{2\ell} Z_t\right| \ge h_\ell \frac{b\ell}{n}
\left| \sum_{t=1}^{2\ell} Z_t =
\frac{2b\ell}{n}\right.\right)\\
& \le& 2\exp{\left(-\frac{\epsilon^2 b}{2+h_l}\right)}\le
\frac{\delta}{2}
\end{eqnarray*}
where $\delta=\frac{\epsilon}{m\cdot n^m\cdot E}$.

For the second term of (\ref{twoterms}), we can define the same
$Z_t$, and we have
\begin{eqnarray*}
P\left(\left|\sum_{t=1}^{\ell} Y_t - \frac{1}{2}\sum_{t=1}^{2\ell}
Y_t\right| \ge \frac{h_\ell}{2} \frac{b\ell}{n}, \
\sum_{t=1}^{2\ell}Y_t\le\frac{2b\ell}{n}\right) & \le &
P\left(\left|\sum_{t=1}^{\ell} Z_t - \frac{1}{2}\sum_{t=1}^{2\ell}
Z_t\right| \ge \frac{h_\ell}{2} \frac{b\ell}{n}, \
\sum_{t=1}^{2\ell}Z_t=\frac{2b\ell}{n}\right) \\
& \le & P\left(\left|\sum_{t=1}^{\ell} Z_t -
\frac{1}{2}\sum_{t=1}^{2\ell} Z_t\right| \ge \frac{h_\ell}{2}
\frac{b\ell}{n} \left|
\sum_{t=1}^{2\ell}Z_t=\frac{2b\ell}{n}\right.\right)\\
& \le &
2\exp\left(-\frac{\epsilon^2b}{8+2h_\ell}\right)\le\frac{\delta}{2}
\end{eqnarray*}
where the second to last step is due to Lemma \ref{HB} and the last
step holds because $h_{\ell} \le 1$ and the condition made on $B$.

Lastly, we define two prices to be {\it distinct} the same way as we
do in the proof of Lemma 2.2. Then we take a union bound over all
the $n^m$ distinct prices, $i=1,\ldots, m$, and $E$ values of
$\ell$, the lemma is proved. \hfill $\Box$

\subsection{Proof of inequality (\ref{eq:deviation-inequality2})}
\label{app:deviation-inequality2} The proof is very similar to the
proof of Lemma \ref{prop:constraints}. Fix $\bp$, $\ell$ and $i \in
\{1, \ldots, m\}$, we define ``bad" permutations for $\bp, i, \ell$
as those permutations such that all the following conditions hold:
(a) $\bp=\hat{\bp}^{\ell}$, that is, $\bp$ is the price learned as
the optimal dual solution for (\ref{lstep}), (b) $p_i >0$, and (c)
$\sum_{t=1}^{2\ell} a_{it} x_t(\bp)
\le(1-2h_{\ell}-\epsilon)\frac{2\ell}{n} b_i$. We will show that the
probability of bad permutations is small.

Define $Y_t=a_{it}x_t(\V{p})$. If $\V{p}$ is an optimal dual
solution for (\ref{lstep}), and $p_i>0$, then by the KKT conditions
the $i^{th}$ inequality constraint holds with equality. Therefore,
by Lemma \ref{KKT}, we have:
\begin{eqnarray*}
\sum_{t=1}^{\ell} Y_t=\sum_{t=1}^{\ell} a_{it} x_t({\V{p}}) \ge
(1-h_{\ell})\frac{\ell}{n}b_i - m \ge
(1-h_{\ell}-\epsilon)\frac{\ell}{n}b_i,
\end{eqnarray*}
where the last inequality follows from $B=\min_i b_i \ge
\frac{m}{\epsilon^2}$, and $\ell \ge n\epsilon$. Therefore, the
probability of ``bad" permutations for $\bp, i, \ell$ is bounded by:
\begin{eqnarray*}
&& P \left(\sum_{t=1}^{\ell} Y_t \ge (1-h_{\ell} -\epsilon)
\frac{\ell}{n}b_i, \sum_{t=1}^{2\ell} Y_t \le (1-
2h_{\ell}-\epsilon)\frac{2\ell}{n}b_i\right) \\
 &\le & P\left(\left|\sum_{t=1}^{\ell} Y_t - \frac{1}{2}\sum_{t=1}^{2\ell}
Y_t\right| \ge h_\ell
\frac{b_i\ell}{n}\left|\sum_{t=1}^{2\ell}Y_t\le
(1-2h_\ell-\epsilon)\frac{2\ell}{n}b_i\right.\right)\\
&\le & P\left(\left|\sum_{t=1}^{\ell} Z_t -
\frac{1}{2}\sum_{t=1}^{2\ell} Z_t\right| \ge h_\ell
\frac{b_i\ell}{n}\left|\sum_{t=1}^{2\ell}Z_t =
(1-2h_\ell-\epsilon)\frac{2\ell}{n}b_i\right.\right)\\
 &\le &
2\exp{\left(-\frac{\epsilon^2 b_i}{2}\right)}\le \delta,
\end{eqnarray*}
where $Z_t = \frac{(1-2h_\ell - \epsilon)2\ell b_i
Y_t}{n\sum_{t=1}^{2\ell} Y_t}$ and $\delta=\frac{\epsilon}{m\cdot
n^m\cdot E}$. The last inequality follows from the condition on $B$.
Next, we take a union bound over all the $n^m$ distinct  $\V{p}$'s,
$i=1,\ldots, m$, and $E$ values of $\ell$, we conclude that with
probability $1-\epsilon$
$$\sum_{t=1}^{2\ell} a_{it}\hat{x}_t(\hat{\bp}^{\ell})\ge (1-2h_{\ell}-\epsilon)\frac{2\ell}{n}b_i$$
for all $i$ such that $\hat{\bp}_i>0$ and all $\ell$. \hfill $\Box$

\section{Detailed steps for Theorem \ref{th:lowerbound}}
\label{app:lower-bound}
Let $c_1, \ldots, c_n$ denote the $n$ customers. 
For each $i$, the set $R_i \subseteq \{c_1, \ldots, c_n\}$ of
customers with bid vector $\V{w}_i$ and bid value $1$ or $3$ is
fixed with $|R_i| = 2B/z$ for all $i$. Conditional on set $R_i$ the
bid values of customers $\{c_j, j\in R_i\}$ are independent random
variables that take value $1$ or $3$ with equal probability.

Now consider the $t^{th}$ bid of $(2, \V{w_i})$. In at least $1/2$
of the random permutations, the number of bids from set $R_i$ before
the bid  $t$ is less than $B/z$. Conditional on this event, with a
constant probability the bids in $R_i$ before $t$ take values such
that the
bids after $t$ can make the number of $(3, \V{w}_i)$ bids more than
$B/z$ with a constant probability and  less than $B/z-\sqrt{B/4z}$
with a constant probability. This probability calculation is similar
to the one used by \cite{kleinberg} in his proof of the necessity of
condition $B\ge \Omega(1/\epsilon^2)$. For completeness, we derive
it in the Lemma \ref{lem:supp} towards the end of the proof.

Now, in the first type of instances (in which the number of $(3,
\V{w}_i)$ bids are more than $B/z$), retaining a $(2, \V{w}_i)$ bid
is a ``potential mistake" of size $1$; similarly, in the second type
of instances (in which the number of $(3,\V{w}_i$) bids are less
than $B/z$), skipping a $(2, \V{w}_i)$ bid is a potential mistake of
size $1$. We call it a potential mistake of size $1$ because it will
cost a profit loss of $1$ if the online algorithm decides to pick
$B/z$ of $\V{w}_i$ bids. Among these mistakes, $|r_i - B/z|$ of them
may be recovered in each instance by deciding to pick $r_i \ne B/z$
of $\V{w}_i$ bids.

The total expected number of potential mistakes is
$\Omega(\sqrt{Bz})$ (since there are $\sqrt{B/4z}$ of $(2,\V{w}_i)$
bids for every $i$). By Claim \ref{claim:low}, no more than a
constant fraction of instances can recover more than $7\epsilon B$
of the potential mistakes.

Let $\ONLINE$ denote the expected value for the online algorithm
over random permutation and random instances of the problem.
Therefore,
$$ \ONLINE \le \OPT - \Omega(\sqrt{z B}-7\epsilon B).$$
Now, observe that $\OPT \le 7B$. This is because by construction
every set of demand vectors (consisting of either $\V{v}_i$ or
$\V{w}_i$ for each $i$) will have at least $1$ item in common, and
since there are only $B$ units of this item available, at most $2B$
demand vectors can be accepted giving a profit of at most $7B$.
Therefore, $\ONLINE \le \OPT (1- \Omega(\sqrt{z/B}-7\epsilon))$, and
in order to get $(1-\epsilon)$ approximation factor we need
$$\Omega(\sqrt{z/B} - 7\epsilon) \le O(\epsilon) \Rightarrow  B \ge \Omega(z/\epsilon^2).$$
This completes the proof of Theorem \ref{th:lowerbound}.

\begin{lemma}
\label{lem:supp} Consider $2k$ random variables $Y_j, j=1,\ldots,
2k$ that take value $0/1$ independently with equal probability.
Let $r \le k$.
Then with constant probability $Y_1, \ldots, Y_r$ take value such
that $\sum_{j=1}^{2k} Y_j$ can be greater or less than its expected
value $k$ by $\sqrt{k}/2$ with equal constant probability.
$$P\left(\left.\left|\sum_{j=1}^{2k} Y_j - k\right|\ge \lceil \sqrt{k}/2 \rceil \ \right| \ Y_1, \ldots Y_r\right) \ge c$$
for some constant $0<c <1$.
\end{lemma}
{\noindent\bf Proof of Lemma \ref{lem:supp}:}
\begin{itemize}
\item Given $r\le k$, $|\sum_{j\le r} Y_j -r/2| \le \sqrt{k}/4$ with constant probability (by central limit theorem).
\item Given $r\le k$, $|\sum_{j>r} Y_j - (2k-r)/2)| \ge 3\sqrt{k}/4$ with constant probability.
\end{itemize}
Given the above events $|\sum_{j} Y_j -k|\ge \sqrt{k}/2$, and by
symmetry both events have equal probability. \hfill $\Box$
\section{Online multi-dimensional linear program}
\label{app:multi_dimension}

\subsection{Proof of Lemma \ref{multilemma1}}
Using Lagrangian duality, observe that given optimal dual solution
$\bp^*$, optimal solution $\bx^*$ is given by:
\begin{equation}
\begin{array}{lll}
\rm maximize & \V{f}^T_t\bx_t - \sum_i p^*_i\V{g}_{it}^T\bx_t & \\
\mbox{subject to} & \V{e}^T\bx  \le 1, \bx_t \ge 0.
\end{array}
\end{equation}
Therefore, it must be true that if $x^*_{tr} = 1$, then $r \in \arg
\max_{j} \{f_{tj} - (\bp^*)^T\V{g}_{tj}\}$ and $f_{tr} -
(\bp^*)^T\V{g}_{tr} \ge 0$ This means that for $t$'s such that
$\max_{j} \{f_{tj} - (\bp^*)^T\V{g}_{tj}\}$ is strictly positive and
$\arg \max_{j}$ returns a unique solution,  $\V{x}_t(\bp^*)$ and
$\V{x}^*_t$ are identical. By random perturbation argument there can
be at most $m$ values of $t$ that does not satisfy this condition
(for each such $t$, $\bp$ satisfies an equation $f_{tj} -
\bp^T\V{g}_{tj} = f_{tl} - \bp^T\V{g}_{tl}$ for some $j,l$, or
$f_{tj} - \bp^T\V{g}_{tj}  = 0$ for some $j$). This means $\bx^*_t$
and $\bx_t(\bp^*)$ differs for at most $m$ values of $t$. \hfill
$\Box$

\subsection{Proof of Lemma \ref{multilemma2}}

Consider $nk^2$ expressions
$$\begin{array}{ll}
f_{tj} - \bp^Tg_{tj} - (f_{tl} - \bp^Tg_{tl}),&  1\le j,l\le k, j\ne l, 1\le t\le n \\
f_{tj} - \bp^Tg_{tj}, 1 \le j \le k, & 1\le t\le n.
\end{array}
$$
$\V{x}_t(\bp)$ is completely determined once we determine the subset
of expressions out of these $nk^2$ expressions that are assigned a
non-negative value. By theory of computational geometry, there can
be at most $(nk^2)^m$ such distinct assignments. \hfill $\Box$

\end{APPENDICES}

\theendnotes

\ACKNOWLEDGMENT{The authors thank the two anonymous referees and the
associate editor for their insightful comments and suggestions.}

\bibliographystyle{ormsv080} 
\bibliography{ref} 

%
%
%
%
%

\end{document}